\pdfoutput=1
\documentclass[final,twocolumn]{IEEEtran}
\usepackage{amsmath,amsfonts,amssymb,epsfig,cite,bm,pstricks,multirow,subfigure,setspace}

\title{``Rewiring'' Filterbanks for Local Fourier Analysis: Theory and Practice} 

\author{Keigo Hirakawa,~\IEEEmembership{Member,~IEEE}, and Patrick J. Wolfe,~\IEEEmembership{Senior~Member,~IEEE}
\thanks{Based upon work supported in part by the National
Science Foundation under Grant No.~DMS-0652743.
        \newline \indent
K.~Hirakawa is with the Intelligent Signal Systems Laboratory, University of Dayton, College Park, Dayton, OH 45469 (e-mail: k.hirakawa@notes.udayton.edu);
and P.~J.~Wolfe is with the Statistics and
Information Sciences Laboratory, Harvard University, Oxford Street,
Cambridge, MA 02138 (e-mail: patrick@seas.harvard.edu).  
This work was completed while K.~Hirakawa
was a Research Associate at Harvard University.}
}

\newtheorem{theorem}{Theorem}
\newtheorem{property}{Property}[section]
\newtheorem{proposition}{Proposition}[section]
\newtheorem{corollary}{Corollary}[section]

\newtheorem{example}{Example}[section]
\newtheorem{definition}{Definition}[section]
\newtheorem{remark}{Remark}[section]

\newcommand{\qed}{\nobreak \ifvmode \relax \else
      \ifdim\lastskip<1.5em \hskip-\lastskip
      \hskip1.5em plus0em minus0.5em \fi \nobreak
      \vrule height0.75em width0.5em depth0.25em\fi}

\begin{document}

\maketitle

\begin{abstract}

This article describes a series of new results outlining equivalences between certain ``rewirings'' of filterbank system block diagrams, and the corresponding actions of convolution, modulation, and downsampling operators.  This gives rise to a general framework of reverse-order and convolution subband structures in filterbank transforms, which we show to be well suited to the analysis of filterbank coefficients arising from subsampled or multiplexed signals. These results thus provide a means to understand \emph{time-localized} aliasing and modulation properties of such signals and their subband representations---notions that are notably absent from the global viewpoint afforded by Fourier analysis. The utility of filterbank rewirings is demonstrated by the closed-form analysis of signals subject to degradations such as missing data, spatially or temporally multiplexed data acquisition, or signal-dependent noise, such as are often encountered in practical signal processing applications.

\end{abstract}

\section{Introduction}
\label{sec:introduction}

Since the earliest days of signal and waveform analysis, engineers have recognized the wide utility of parameterized families of \emph{filters}: convolution operators that are directly represented by finite-length sequences of real numbers.  Parallel banks of such finite-impulse-response filters, including short-time Fourier and wavelet transforms, have long been a canonical tool for analyzing signals, images, and other data sets that arise in a variety of applications across scientific fields\cite{ref:Mallat_1999,ref:Strang_1996}. The purpose of this article is to further expand filterbank theory and practice by developing a general framework of reverse-order and convolution subband structures in filterbank transforms.  It describes a series of new results outlining equivalences between certain ``rewirings'' of filterbank block diagrams, and the ``localized'' aliasing and modulation properties of sampled signals and their subband representations, which we describe below.

Sampled signals are typically acquired as linear functionals of the underlying data object of interest, which in turn is defined with respect to a continuous variable such as time or space. The actions of the convolution operators that comprise filterbanks are studied through their Fourier transforms, under the correspondence of element-wise multiplication on the dual group. Since sampling a continuous-time function periodizes its Fourier transform, however, care must be taken that no information is lost in the process. Indeed, in the absence of additional assumptions, it is not in general possible to recover signals that have \emph{aliased}; that is, signals whose Fourier transforms are supported on intervals so large that this periodization mixes distinct Fourier coefficients.

As the bandwidth of any function is directly determined by its \emph{global} smoothness, Fourier analysis does not lend itself to a meaningful analysis of signals whose smoothness varies and hence are not low-pass \emph{everywhere}. In contrast, parallel banks of convolution operators with finite support are fundamentally \emph{local} in nature.  It is well known, for instance, that the flexibility afforded by wavelets to adapt to the local regularity of functions is essential in yielding the sparsity properties necessary for effective signal and image analysis, as well as contemporary signal acquisition techniques such as compressed sensing\cite{ref:Donoho_2006,ref:Candes_2006a,ref:Candes_2006b,ref:Candes_2006c,ref:Haupt_2006,ref:Eldar_2009}.

Filterbanks are hence essential engineering tools for data analysis.  However, definitions of aliasing and frequency modulation in the global, Fourier context preclude the closed-form filterbank analysis of signals subject to missing data, spatially or temporally multiplexed acquisition, or signal-dependent noise effects. While it is well known how to apply filterbanks to analyze, modify, and enhance signals that are free from aliasing or modulatory effects, the literature presently lacks a unified filterbank theory for these settings.

In this work we develop a set of results necessary to fully understand and exploit the \emph{local} aliasing and modulation properties of sampled signals and their subband representations. Though our motivation stems from signal processing problems typically encountered in practice (such as those mentioned above, to which we return at the end of the article), our results are more general, showing equivalences between certain ``rewirings'' of filterbank block diagrams, and the corresponding actions of convolution, modulation, and downsampling operators. Our primary contributions are the introduction and analysis of two cardinal rewiring mechanisms---reverse-order subband structure (ROSS) and subband convolution structure (SCS)---by which filterbank subbands are coupled together to describe the relationship between localized aliasing, modulation, and convolution.

The framework we introduce is distinct from work involving signal recovery methods \cite{ref:Donoho_2006,ref:Candes_2006a,ref:Candes_2006b,ref:Candes_2006c,ref:Haupt_2006,ref:Eldar_2009} and sampling theorems \cite{ref:Walter_1992,ref:Xia_1993,ref:Vaidyanathan_2001,ref:Vaidyanathan_1988} 
in the extant literature. Such work has successfully characterized sufficient
conditions for exact reconstruction when filterbank
theory is used to restrict the class of signals under consideration, or to specify the fundamental compressibility of its members.  In
contrast, this article employs filterbank theory
to describe the data acquisition and sampling process itself, rather
than any properties of a given signal class.  Notions of localized aliasing and
localized modulation are
intimately connected with the ROSS and SCS analyses that we introduce below, and are also complementary to other well-understood concepts in filterbank analysis.

The article is organized as follows.  In the remainder of Section I we introduce key definitions and filterbank notation, and provide a simple example of local aliasing and local modulation that motivates our subsequent analysis. In
Section \ref{sec:ROSS} we introduce our first ``rewiring''
notion---that of reverse-order subband structure---and derive
corresponding expressions for the filterbank
coefficients corresponding to a subsampled signal.  In Section \ref{sec:SCS} we
build on this work to introduce the notion of
subband convolution structure---our second means of filterbank
rewiring---and show how it leads to a convolution
theorem particularly suited to the local modularity of the filterbank
transform. We conclude with a discussion in Section \ref{sec:likelihood} where we consider the practical use of these two notions in problems involving missing data, multiplexed signal acquisition, and signal-dependent noise.

\subsection{Key Definitions and Filterbank Notation}
\label{sec:time-frequency}
\begin{figure}
\subfigure[$x$]{\includegraphics[width=.158\textwidth]{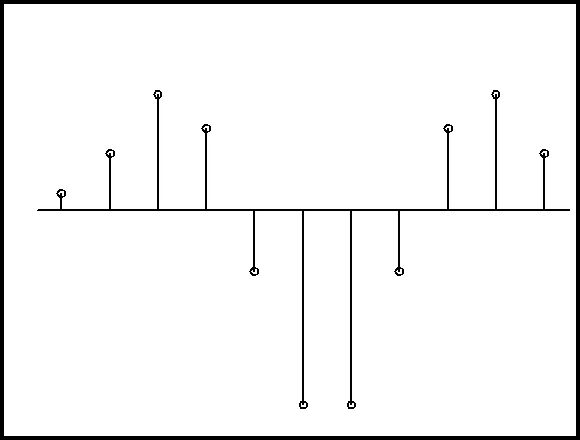}}
\subfigure[$x_m$]{\includegraphics[width=.158\textwidth]{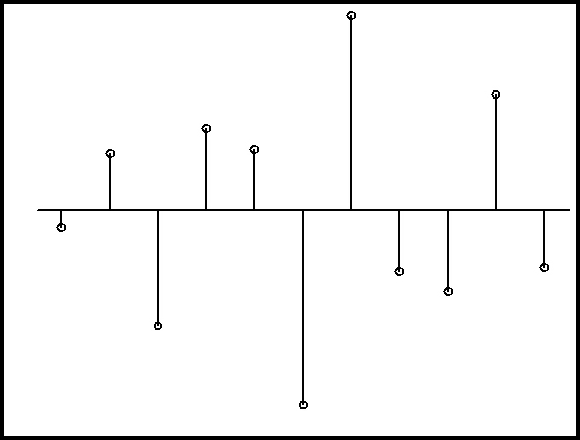}}
\subfigure[$x_s$]{\includegraphics[width=.158\textwidth]{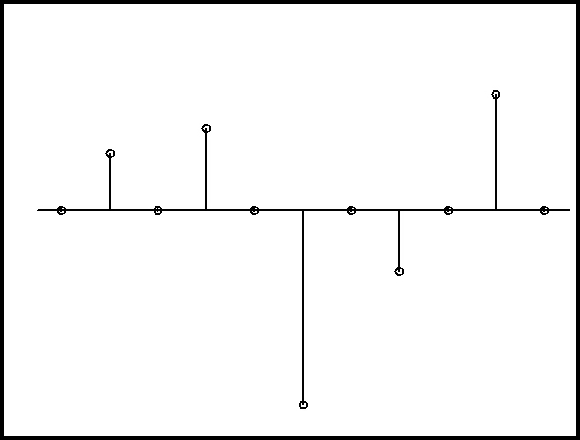}}
\subfigure[$\hat{x}$]{\includegraphics[width=.158\textwidth]{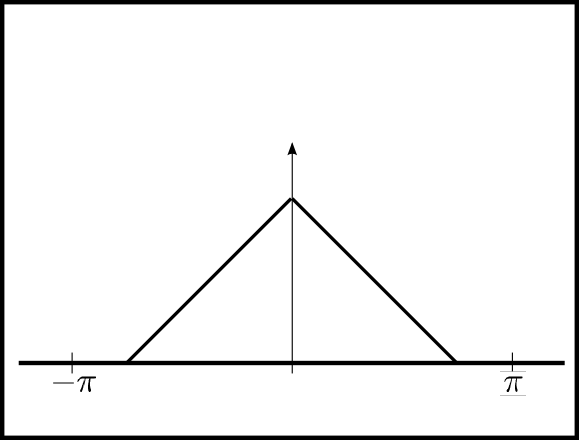}}
\subfigure[$\hat{x}_m$]{\includegraphics[width=.158\textwidth]{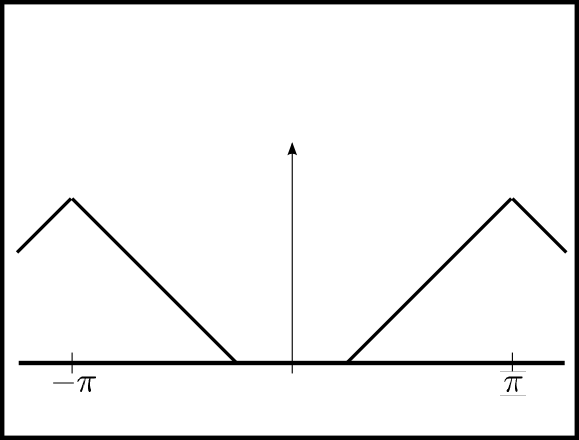}}
\subfigure[$\hat{x}_s$]{\includegraphics[width=.158\textwidth]{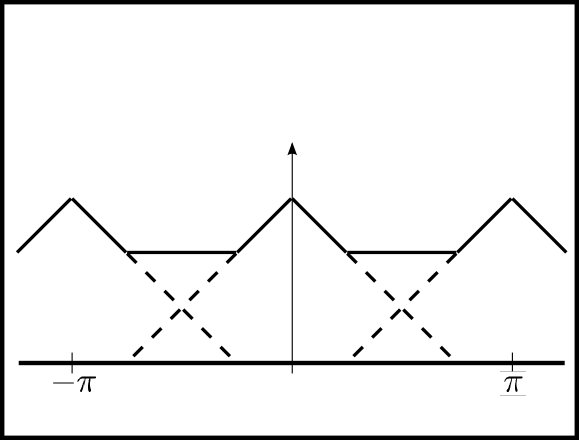}}
\caption{Pictorial illustrations of modulation and sampling in both the time (a-c) and frequency (d-f) domains.}
\label{fig:sampling}
\end{figure}
 
Throughout, let $x\in\ell^2(\mathbb{Z})$ be a real-valued sequence indexed by
$n\in\mathbb{Z}$. \emph{Subsampling} is the operation of replacing every odd-numbered element of $x$ by
zero, and hence the \emph{subsampled sequence} $x_s$ is defined element-wise as
\begin{align*}
x_s[n]&:=\begin{cases}x[n]&\text{if $n$ even,}\\0&\text{if $n$
    odd.}\end{cases}
\end{align*}
(Note that this is distinct from \emph{downsampling}, a dilation of the index set of
$x$ to yield $x[2n]$, such that odd-numbered samples are ``dropped'' and only
even-numbered ones retained.) 
Equivalently, $x_s[n]$ is an arithmetic average
of $x[n]$ and its frequency-\emph{modulated} version $x_m[n]:=(-1)^nx[n]\text{.}$:
\begin{align*}
x_s[n]=\frac{1}{2}\left(x[n]+x_m[n]\right)\text{.}
\end{align*}
Figures \ref{fig:sampling}(a)-(c) serve as a reminder to illustrate how
samples in $x$ and $x_m$ with opposite signs cancel out to yield
$x_s$; we shall frequently refer back to them later.

Let $\hat{x}$ denote the discrete-time Fourier transform of $x$, with $\omega \in \mathbb{R}/2\pi$ 
its corresponding normalized angular frequency. Then it follows that 
$\hat{x}_m(\omega)=\hat{x}(\omega-\pi)$ and
\begin{align}\label{eqn:FFT_aliasing}
\hat{x}_s(\omega)&=\frac{1}{2}\Big[\hat{x}(\omega)+\hat{x}(\omega-\pi)\Big]\text{.}
\end{align}
Here we see that when the bandwidth of $x$---i.e.,~the support of
$\hat{x}$---is sufficiently large, $\hat{x}(\omega)$ and 
$\hat{x}_m(\omega)$ are indistinguishable in $\hat{x}_s(\omega)$; as shown in Figures \ref{fig:sampling}(d)-(g), their supports
overlap in the Fourier domain. This phenomenon is called
\emph{aliasing}; in the absence of
additional information, aliased portions of $x$ cannot be recovered
from $\hat{x}_s(\omega)$ alone. 

\begin{figure}
\centering
\subfigure[One-level filterbank structure]{\includegraphics[scale=.25]{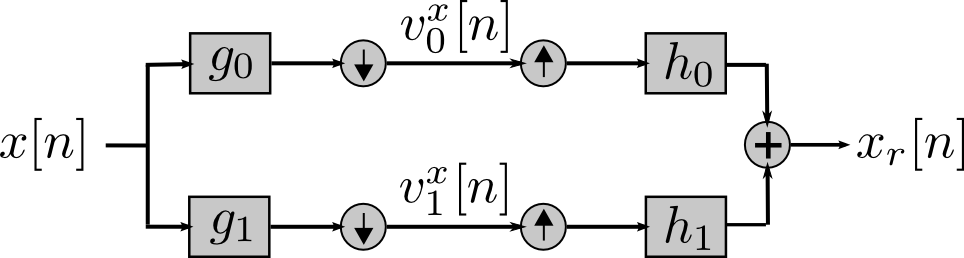}}
\subfigure[One-level complementary filterbank structure]{\includegraphics[scale=.25]{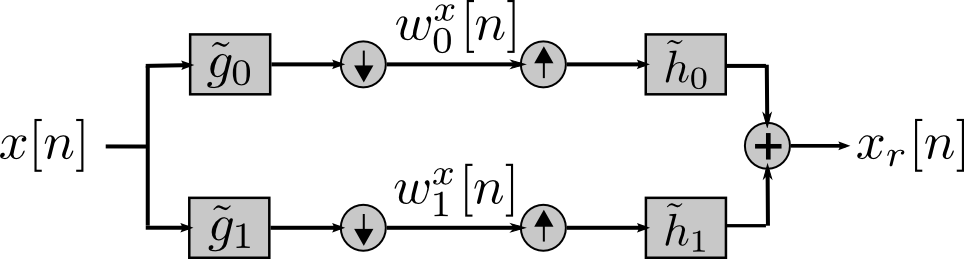}}
\caption{
One-level filterbank (a) and its complementary structure (b).
Diagram (a) represents analysis and synthesis filters $\{g_0,g_1\}$
and $\{h_0,h_1\}$, respectively, as well as filterbank coefficient
sets $\{v_0^x,v_1^x \}$ resulting from the action of the analysis
convolution operator and subsequent downsampling on a sequence $x$.
Reversing the downsampling procedure and applying the synthesis filters yields a ``reconstructed'' sequence $x_r$. 
Diagram (b) represents filters $\tilde{g}_i, \tilde{h}_i$ ``complementary'' to (a), obtained by swapping the roles of $g$ and $h$, respectively, and then applying an affine transformation.}
\label{fig:filterbank00}
\end{figure}

The Fourier transform is a fundamentally global operation; modulation and aliasing mix non-local information from the sequence $x$.  In contrast, a \emph{filterbank} maps a sequence $x$ to some alternative representation by way
of localized \emph{filterbank coefficients} (``analysis''), and subsequently yields a linear reconstruction $x_r$ (``synthesis''); well-known examples include short-time Fourier and wavelet representations. As the analysis operator acts
linearly, we write its action as an inner product, calling it a filter
when it is a convolution operator parameterized by translation along a
sublattice of $\mathbb{Z}$, as is the case considered here.
If this convolution operator is represented by the actions of a
finite-length, real-valued sequences in $\ell^2(\mathbb{Z})$, then we refer to
a real-valued, finite-impulse-response analysis filter $g$ and corresponding synthesis filter $h$. Figure
\ref{fig:filterbank00}(a) illustrates a basic filterbank structure, with
two analysis filters $\{g_0,g_1\}$ and two synthesis filters $\{h_0,h_1\}$. 
We denote by $\{v_0^x,v_1^x\}$ the corresponding filterbank coefficient sequences, defined as follows.
\begin{definition}[Filterbank Coefficient Sequence]\label{def:v}
We call $v^x_i\in\ell^2(\mathbb{Z})$ a one-level filterbank
coefficient sequence corresponding to $x$ if
\begin{align}\label{eqn:filterbank}
v^x_i[n]:=\Big(g_i[m]\star_m x[m]\Big)[2n]\text{,}
\end{align}
where the summation in the discrete convolution $\star_m$ is
performed over the index $m$, and the subsequent notion of downsampling by two is
reflected by the index set $\{2n:n\in\mathbb{Z}\}$.
\end{definition}
This composition of convolution and
dilation implies in turn that
\begin{align}\label{eqn:filterbank_FFT}
\hat{v}^x_i(\omega)=\frac{1}{2}\left[\hat{g}_i\left(\frac{\omega}{2}\right)\hat{x}\left(\frac{\omega}{2}\right)+
\hat{g}_i\left(\frac{\omega}{2}-\pi\right)\hat{x}\left(\frac{\omega}{2}-\pi\right)\right]\text{,}
\end{align}
and the set of transform coefficients $\{v^x_i[n]\}_{n\in\mathbb{Z}}$ is collectively
referred to as the $i$th filterbank subband. 
Typically $g_0$ and $h_0$ are smooth (i.e.,~low-pass) filters, while
$g_1$ and $h_1$ have zero average. Thus, $v^x_0$
provides a measure of local low-frequency energy concentration, while
$v^x_1$ captures local 
high-frequency energy, with temporal localization provided by the
finite support of $\{g_i,h_i\}_{i\in\mathbb{Z}_2}$. A
filterbank's joint time-frequency resolution can be fine-tuned by 
recursively nesting copies of the basic one-level transform structure
illustrated in Figure \ref{fig:filterbank00}(a), yielding the multi-level filterbank structures that we consider later in Section \ref{sec:multi-level}. 

Note that the Fourier representation of~\eqref{eqn:filterbank_FFT}
implies the superposition of shifted copies of the resultant filtered
spectra, which in general will give rise to aliasing of the type
illustrated in Figure \ref{fig:sampling}(f).  It is thus natural to ask
for conditions under which this aliasing will cancel---a prerequisite for the exact reconstruction of any input sequence $x \in
\ell^2(\mathbb{Z})$ from its filterbank coefficients, such that $x_r=x$ in the diagram of Figure \ref{fig:filterbank00}(a). To this end we 
arrive at the following well-known definition, which stems from \emph{global} properties of the Fourier transform. 

\begin{definition}[Perfect Reconstruction Filterbank] 
\label{def:prf}
A \emph{perfect reconstruction filterbank}
$\{g_i,h_i\}_{i\in\mathbb{Z}_2}$ admits for all $x\in\ell^2(\mathbb{Z})$ the relation
\begin{align}\label{eqn:inverse_filterbank}
\hat{x}_(\omega):=\hat{h}_0(\omega) \hat{v}^x_0(2\omega)+\hat{h}_1(\omega) \hat{v}^x_1(2\omega)=\hat{x}(\omega)\text{.}
\end{align}
Equivalently, as shown in Figure \ref{fig:filterbank00}(a), we have for all $n\in\mathbb{Z}$ that
\begin{align*}
x_r[2n]=&\Big(h_0[2m]\star_m v^x_0[m]\Big)[n]\\
&+\Big(h_1[2m]\star_m v^x_1[m]\Big)[n]=x[2n]\\
x_r[2n+1]=&\Big(h_0[2m+1]\star_m v^x_0[m]\Big)[n]\\
&+\Big(h_1[2m+1]\star_m v^x_1[m]\Big)[n]=x[2n+1]\text{.}
\end{align*}
\end{definition}
As described earlier, the sequences of operations corresponding to the forward transform step in \eqref{eqn:filterbank} and the reconstruction step in \eqref{eqn:inverse_filterbank} are commonly 
referred to as the \emph{analysis} and \emph{synthesis filterbanks},
respectively.
\begin{remark}[Haar Filterbank Transform] \label{rem:HFT}
Perhaps the most well-known example of a perfect reconstruction
filterbank is given by the so-called Haar transform, which may be
defined in terms of its $z$-transform as:
\begin{align}\label{eqn:HFT}
\begin{split}
\sum_n g_i[n]z^{-n}=&1+(-1)^iz\text{,}\\
\sum_n h_i[n]z^{-n}=&\frac{1}{2}\Big[(-1)^i+z^{-1}\Big]\text{.} 
\end{split}
\end{align}
Note that $g_0$ and $g_1$ in~\eqref{eqn:HFT} amount to the sum and difference of
neighboring samples, respectively, and it is clear that the original
sequence $x[n]$ is easily recoverable from the corresponding sequences
of filterbank coefficients. 
\end{remark}

Important and well-known results associated with Definition \ref{def:prf} established
the following (see, e.g.,~\cite{ref:Mallat_1999}) .
\begin{property}[Alias Cancellation Equivalence]\label{prpt:Vetterli} 
The set $\{g_i,h_i\}_{i\in\mathbb{Z}_2}$ is a perfect 
reconstruction filterbank if and only if
\begin{align}
2=\hat{g}_0(\omega)\hat{h}_0(\omega)&+\hat{g}_1(\omega)\hat{h}_1(\omega)\label{eqn:vetterli1}\\
0=\hat{g}_0(\omega)\hat{h}_0(\omega-\pi)&+\hat{g}_1(\omega)\hat{h}_1(\omega-\pi)\label{eqn:vetterli2}
\end{align}
\end{property}
\begin{property}[Analysis-Synthesis Symmetry]\label{prpt:Mallat}
If finite-impulse-response filters $\{g_i,h_i\}_{i\in\mathbb{Z}_2}$ comprise a perfect
reconstruction filterbank, then there exist $a\in\mathbb{R}\setminus \{0\}$ and
$b\in\mathbb{Z}$ such that
\begin{align*}
\hat{g}_i(\omega)&=(-1)^{1-i}ae^{j(2b+1)\omega}\hat{h}_{1-i}(\omega-\pi)\text{.}
\end{align*}
\end{property}
\begin{remark}
The condition of \eqref{eqn:vetterli2} in Property \ref{prpt:Vetterli} guarantees that the aliased components in $v^x_0$ and $v^x_1$ 
cancel, so the impulse response
of the overall filterbank structure in Figure
\ref{fig:filterbank00}(a) is equal to one-half of the expression of
\eqref{eqn:vetterli1}; i.e., everywhere constant and equal to unity.
Property \ref{prpt:Mallat} makes explicit the fact that an
analysis filterbank uniquely defines its corresponding synthesis filterbank.
These properties serve as the foundation for
the reverse-order subband filterbank structure that we introduce in
Section \ref{sec:ROSS} below.
\end{remark}

\subsection{Motivating Example for Filterbank ``Rewiring''}
\label{sec:toy}

The preceding section has served to introduce the basic notions of
filterbank theory that we shall employ here.  Before continuing, it
is instructive to consider a simple motivating example based on the
simplest case of the Haar filterbank transform.
In essence, we will see that ``rewiring'' filterbank diagrams such as
those in Figure \ref{fig:filterbank01} can be related to the actions of convolution, modulation, and downsampling operators. In subsequent sections we develop these properties formally, and show how they yield new insights into important practical problems.

To begin, consider the symmetric Hadamard matrix
\begin{equation*}
\bm{\Phi} = \begin{bmatrix}1 &1\\1 &-1\end{bmatrix}
\end{equation*}
which maps two-dimensional vectors to the corresponding sums and
difference of their components, and thus serves to define the
(unnormalized) one-level Haar filterbank.  Note that $\bm{\Phi}^{-1} = \frac{1}{2} \bm{\Phi}$,
and consider two systems of linear equations in $\bm{\Phi}$ that will serve to illustrate the concepts of reverse-order and convolution subband structure: 
\begin{align*}
\begin{bmatrix}q\\r\end{bmatrix}=&\bm{\Phi}\begin{bmatrix}a\\b\end{bmatrix},\quad
\begin{bmatrix}s\\t\end{bmatrix}=\bm{\Phi}\begin{bmatrix}c\\d\end{bmatrix}\text{.}
\end{align*}
\begin{example}[Reverse-Ordering and Subsampling]
Consider the first system of equations above, and suppose that we replace $b$ with $-b$, yielding a ``modulated'' version of $[a,b]^T$. We then observe that the Haar
transform of $[a,-b]^T$ results in a \emph{reverse-ordering} of $q$ and $r$, which play the roles of low-pass and high-pass components, respectively:
\begin{align*}
\begin{bmatrix}1&1\\1&-1\end{bmatrix}\begin{bmatrix}a\\-b\end{bmatrix}
=&\frac{1}{2}\begin{bmatrix}1&1\\1&-1\end{bmatrix}
\begin{bmatrix}1&1\\-1&1\end{bmatrix}\begin{bmatrix}q\\r\end{bmatrix}
=\begin{bmatrix}r\\q\end{bmatrix}\text{.}
\end{align*}
Since summing $[a,b]^T$ and its modulated version corresponds to subsampling, we next compute the Haar transform of $[a,0]^T$, and observe that this results in an arithmetic averaging of $q$ and $r$:
\begin{align*}
\begin{bmatrix}1&1\\1&-1\end{bmatrix}\begin{bmatrix}a\\0\end{bmatrix}
=&\frac{1}{2}\begin{bmatrix}1&1\\1&-1\end{bmatrix}\left(\begin{bmatrix}a\\b\end{bmatrix}+\begin{bmatrix}a\\-b\end{bmatrix}\right)=\frac{1}{2}\begin{bmatrix}q+r\\q+r\end{bmatrix}\text{.}
\end{align*}
We see from this simple example that the ``swapping'' and the ``combining'' of \emph{low-pass} and
\emph{high-pass} components are reminiscent of modulation and aliasing in
the traditional Fourier sense, as illustrated respectively in Figures \ref{fig:sampling}(e) and \ref{fig:sampling}(f).
\end{example}

\begin{example}[Convolution and Pointwise Multiplication]
Now consider the element-wise product of the vectors $[a,b]^T$ and $[c,d]^T$. 
The Haar transform of this product $[ac,bd]^T$ is:
\begin{align*}
\begin{bmatrix}1&1\\1&-1\end{bmatrix}\begin{bmatrix}ac\\bd\end{bmatrix}
=&\frac{1}{4}\begin{bmatrix}1&1\\1&-1\end{bmatrix}\begin{bmatrix}(q+r)(s+t)\\(q-r)(s-t)\end{bmatrix}\\
=&\frac{1}{4}\begin{bmatrix}(q+r)(s+t)+(q-r)(s-t)\\(q+r)(s+t)-(q-r)(s-t)\end{bmatrix}\\
=&\frac{1}{2}\begin{bmatrix}qs+rt\\qt+sr\end{bmatrix}\text{.}
\end{align*}
The symmetry of $qs+rt$ and $qt+sr$ suggests a kind of cyclic convolution of $[q,r]^T$ and $[s,t]^T$. 
In fact, we will see in Section \ref{sec:SCS} that our filterbank rewiring techniques recover precisely this notion of group structure, in direct analogy to \emph{global} Fourier analysis.
In Section \ref{sec:likelihood}, these ideas will reappear in the context of analysis of signals subject to multiplicative noise corruption.
\end{example}

\section{Reverse-Order Subband Structure and Localized Aliasing}
\label{sec:ROSS}

Having introduced the necessary definitions and given two brief examples, we now begin our technical development of filterbank ``rewiring.'' 
Bearing in mind the examples considered above, we introduce in Section
\ref{sec:complementary} below the notion of \emph{complementary
  filterbanks}, and then employ them to obtain the following results
in Section~\ref{sec:ROSS_theorems}: the \emph{reverse-ordering of 
subband structure} that results from modulation, and \emph{localized aliasing} that
stems from averaging the low- and high-frequency filterbank
coefficients. In Section~\ref{sec:multi-level} we extend these results
to the setting of multi-level filterbanks. 

\subsection{Complementary Filterbanks}
\label{sec:complementary}

\begin{definition}
\emph{(Complementary Filterbanks and Filterbank Coefficients):}\label{def:w}
Let $\{g_i,h_i\}_{i\in\mathbb{Z}_2}$ be a perfect reconstruction
filterbank. Then we define the \emph{complementary filterbank} $\{\tilde{g}_i,\tilde{h}_i\}_{i\in\mathbb{Z}_2}$ as follows:
\begin{align*}
\tilde{g}_i[n]:=&ah_i[n+(2b+1)]\\
\tilde{h}_i[n]:=&a^{-1}g_i[n-(2b+1)]\text{,}
\end{align*}
where $a$ and $b$ are chosen to satisfy Property \ref{prpt:Mallat} of perfect reconstruction filterbank, and we call $w_i^x[n]$ a one-level \emph{complementary filterbank coefficient} corresponding to a sequence $x$ if
\begin{align*}
w^x_i[n]&:=\Big(\tilde{g}_i[m]\star_m x[m]\Big)[2n]\text{.}
\end{align*}
\end{definition}
The following important property of complementary filterbanks follows
directly from Properties \ref{prpt:Vetterli} and \ref{prpt:Mallat} of perfect reconstruction filterbanks.

\begin{proposition}[Complementarity \& Perfect Reconstruction]\label{prop:complementary}
If the set
$\{g_i,h_i\}_{i\in\mathbb{Z}_2}$ is a perfect reconstruction
filterbank, then so is
$\{\tilde{g}_i,\tilde{h}_i\}_{i\in\mathbb{Z}_2}$.
\end{proposition}
\begin{proof}
Appealing to Property \ref{prpt:Mallat}, we see that Fourier transforms of $\tilde{g}_i$ and $\tilde{h}_i$ respectively yield
\begin{align}
\begin{split}
\hat{\tilde{g}}_i(\omega)=&ae^{j(2b+1)\omega}\hat{h}_i(\omega)\label{eqn:complementary_filter}\\
\hat{\tilde{h}}_i(\omega)=&a^{-1}e^{-j(2b+1)\omega}\hat{g}_i(\omega)\text{.}
\end{split}
\end{align}
By substitution, we verify that \eqref{eqn:vetterli1} and \eqref{eqn:vetterli2}
hold for $\{\tilde{g}_i,\tilde{h}_i\}_{i\in\mathbb{Z}_2}$:
\begin{align*}
&\hat{\tilde{g}}_0(\omega)\hat{\tilde{h}}_0(\omega)+\hat{\tilde{g}}_1(\omega)\hat{\tilde{h}}_1(\omega)\\
&\quad=\Big(ae^{j(2b+1)\omega}\hat{h}_0(\omega)\Big)
\Big(a^{-1}e^{-j(2b+1)}\hat{g}_0(\omega)\Big)\\
&\quad\quad +
\Big(ae^{j(2b+1)\omega}\hat{h}_1(\omega)\Big)
\Big(a^{-1}e^{-j(2b+1)}\hat{g}_1(\omega)\Big)
=2\\
&\hat{\tilde{g}}_0(\omega)\hat{\tilde{h}}_0(\omega-\pi)+\hat{\tilde{g}}_1(\omega)\hat{\tilde{h}}_1(\omega-\pi)\\
&\quad=\Big(ae^{j(2b+1)\omega}\hat{h}_0(\omega)\Big)
\Big(a^{-1}e^{-j(2b+1)}\hat{g}_0(\omega-\pi)\Big)\\
&\quad\quad +
\Big(ae^{j(2b+1)\omega}\hat{h}_1(\omega)\Big)
\Big(a^{-1}e^{-j(2b+1)}\hat{g}_1(\omega-\pi)\Big)=0\text{.}
\end{align*}
Hence by Property \ref{prpt:Vetterli}, the set $\{\tilde{g}_i,\tilde{h}_i\}_{i\in\mathbb{Z}_2}$
comprises a perfect reconstruction filterbank.
\end{proof}
Figure \ref{fig:filterbank00}(b) illustrates this complementary filterbank
structure, along with the corresponding complementary coefficients
$w^x_i$. It is natural to ask if a filterbank can be its own complement, and to this end we have the following.
\begin{definition}[Self-Complementary Filterbank]
We call a filterbank $\{g_i,h_i\}_{i\in\mathbb{Z}_2}$
\emph{self-complementary} if
\begin{align}\label{eqn:self-complementary}
v^x_i[n]=(-1)^{1-i}w^x_i[n]\text{.}
\end{align}
\end{definition}
Returning now to Remark \ref{rem:HFT}, we note the following.
\begin{proposition}[Self-Complementarity of Haar Filterbank]\label{rem:self-complementary}
The Haar filterbank is self-complementary.
\end{proposition}
\begin{proof}
It follows from \eqref{eqn:HFT} that the Haar filterbank satisfies the following symmetry:
\begin{align*}
\hat{g}_i(\omega)=(-1)^ie^{j(2b+1)\omega}\hat{g}_i^*(\omega)\text{,}
\end{align*}
with $a=\frac{1}{2}$ and $b=-1$. Applying Property \ref{prpt:Mallat} of perfect reconstruction filterbanks in turn yields
\begin{align}\label{eqn:symmetry}
(-1)^ie^{j(2b+1)\omega}\hat{g}_i^*(\omega)=&(-1)^{1-i}ae^{j(2b+1)\omega}\hat{h}_{1-i}(\omega-\pi)\notag\\
\hat{g}_i^*(\omega)=&-a\hat{h}_{1-i}(\omega-\pi)\text{.}
\end{align}
The well-known identity 
$\hat{g}_i(\omega)=(-1)^ie^{j(2b+1)\omega}\hat{g}_{1-i}^*(\omega-\pi)$
of Smith and Barnwell
\cite{ref:Smith_1984} applies; and upon
substituting this into \eqref{eqn:symmetry}, we obtain the desired result:
\begin{align*}
\hat{g}_i(\omega)=&(-1)^ie^{j(2b+1)\omega}(-a\hat{h}_{i}(\omega))\\
=&(-1)^{1-i} \hat{\tilde{g}}_i(\omega)\text{.}
\end{align*}
\end{proof}
As we show below, complementary filterbanks play a key role in the reverse-ordering of subband structure induced by modulation.

\subsection{Reverse-Order Subband Structure}
\label{sec:ROSS_theorems}


Figure \ref{fig:filterbank01} illustrates the \emph{reversal of subband
ordering} that results when $x$ is modulated by $\pi$ to yield $x_m$: the
low-frequency filterbank coefficient for the modulated signal
($v^{x_m}_0[n]$) behaves like the high-frequency complementary filterbank
coefficient for the original signal ($w^{x}_1[n]$), and
vice-versa. As may be seen by comparing Figure \ref{fig:filterbank01} with Figure \ref{fig:sampling}, this filterbank subband ``role-reversal'' is consistent
with the Fourier interpretation of modulation by $\pi$; in both cases, the low- and high-frequency components are swapped, 
modulo-$2\pi$. We formalize this notion as follows:

\begin{theorem}[Reverse-Order Subband Structure (ROSS)]\label{thm:reverse-order}
If the set $\{g_i,h_i\}_{i\in\mathbb{Z}_2}$ is
a perfect reconstruction filterbank, then
\begin{align}\label{eqn:reverse-order}
v^{x_m}_i[n]=(-1)^iw^x_{1-i}[n]\text{.}
\end{align}
\end{theorem}
\begin{proof}
Modulation of $x$ by $\pi$ implies that we have that
\begin{align*}
\hat{v}^{x_m}_i(\omega)&=\frac{1}{2}\left[\hat{g}_i\left(\frac{\omega}{2}\right)\hat{x}_m\left(\frac{\omega}{2}\right)+
\hat{g}_i\left(\frac{\omega}{2}-\pi\right)\hat{x}_m\left(\frac{\omega}{2}-\pi\right)\right]\\
&=\frac{1}{2}\left[\hat{g}_i\left(\frac{\omega}{2}\right)\hat{x}\left(\frac{\omega}{2}-\pi\right)+
\hat{g}_i\left(\frac{\omega}{2}-\pi\right)\hat{x}\left(\frac{\omega}{2}\right)\right]\text{.}
\end{align*}
By Property \ref{prpt:Mallat} of perfect reconstruction filterbanks and Definition \ref{def:w},
\begin{align}\label{eqn:g_h}
&\hat{g}_i(\omega)=(-1)^{1-i}ae^{j(2b+1)\omega}\hat{h}_{1-i}(\omega-\pi)\notag\\
&\quad=(-1)^i\hat{\tilde{g}}_{1-i}(\omega-\pi)\notag\\
&\hat{v}^{x_m}_i(\omega)\notag\\
&\quad=\frac{(-1)^i}{2}\left[
\hat{\tilde{g}}_{1-i}\left(\frac{\omega}{2}\right)
\hat{x}\left(\frac{\omega}{2}\right)+\hat{\tilde{g}}_{1-i}\left(\frac{\omega}{2}-\pi\right)
\hat{x}\left(\frac{\omega}{2}-\pi\right)
\right]\notag\\
&\quad=(-1)^i\hat{w}^x_{1-i}(\omega)
\text{.}
\end{align}
\end{proof}
\begin{figure}
\centering
\includegraphics[scale=.25]{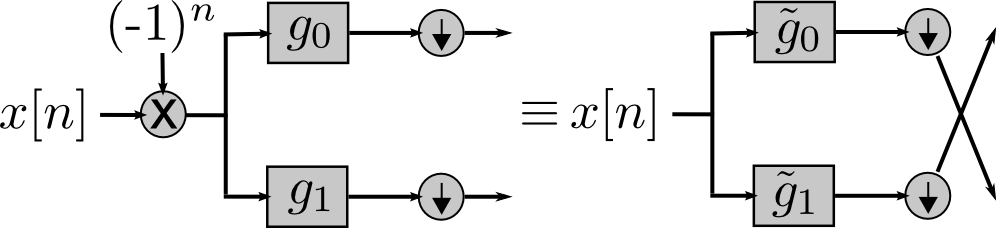}
\caption{Illustration of Theorem \ref{thm:reverse-order}, with the
  left side showing a modulated signal $x_m[n]$ and the right side
  showing reverse-ordering of the complementary filterbank. 
The low- and high-frequency filterbank subbands for the modulated signal
behave like the high- and low-frequency complementary filterbank
subbands for the original signal, respectively.}
\label{fig:filterbank01}
\end{figure}

Applying Property \ref{prpt:Vetterli} of perfect reconstruction filterbanks to \eqref{eqn:g_h} immediately yields the following important corollary.
\begin{corollary}[Modulation induced by ROSS]\label{cor:reverse-order} 
Suppose the set $\{g_i,h_i\}_{i\in\mathbb{Z}_2}$ is a perfect
  reconstruction filterbank. Then 
  \begin{align*}
    0=&\hat{\tilde{g}}_1(\omega)\hat{h}_0(\omega)-\hat{\tilde{g}}_0(\omega)\hat{h}_1(\omega),\\
    2=&\hat{\tilde{g}}_1(\omega)\hat{h}_0(\omega-\pi)-\hat{\tilde{g}}_0(\omega)\hat{h}_1(\omega-\pi)\text{.}
  \end{align*}
\end{corollary}

\begin{figure}
\centering
\includegraphics[scale=.25]{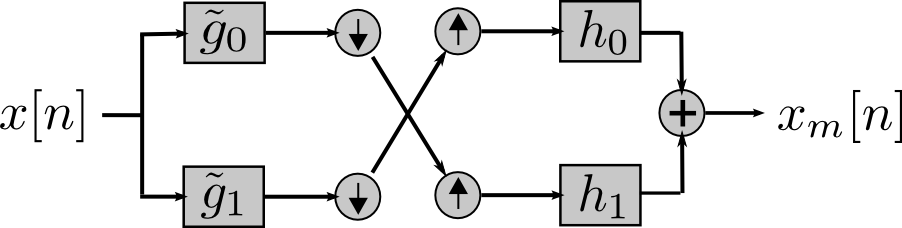}
\caption{Illustration of Corollary \ref{cor:reverse-order}: exchange
  of low- and high-frequency filterbank subbands results in modulation
  (compare to the standard complementary filterbank structure of
  Figure \ref{fig:filterbank00}(b)).}
\label{fig:filterbank02}
\end{figure}
\begin{remark}[Filterbank Interpretation of ROSS Modulation]
An intuitive interpretation of Corollary \ref{cor:reverse-order} is 
that exchanging the low- and high-frequency
filterbank subbands results in modulation.
To see this, consider
reconstruction of the complementary filterbank coefficients with
\emph{reverse-order} subbands, as illustrated in Figure \ref{fig:filterbank02}:
\begin{align*}
\hat{x}_r(\omega)&=\hat{h}_0(\omega) \hat{w}^x_1(2\omega)-\hat{h}_1(\omega)\hat{w}^x_0(2\omega)\\
&=\frac{1}{2}\hat{h}_0(\omega)\Big[\hat{\tilde{g}}_1(\omega)\hat{x}(\omega)+\hat{\tilde{g}}_1(\omega-\pi)\hat{x}(\omega-\pi)\Big]\\
&\quad -\frac{1}{2}\hat{h}_1(\omega)\Big[\hat{\tilde{g}}_0(\omega)\hat{x}(\omega)+\hat{\tilde{g}}_0(\omega-\pi)\hat{x}(\omega-\pi)\Big]\\
&=\frac{1}{2}\hat{x}(\omega)\Big[\hat{h}_0(\omega)\hat{\tilde{g}}_1(\omega)-\hat{h}_1(\omega)\hat{\tilde{g}}_0(\omega)\Big]\\
&\quad +\frac{1}{2}\hat{x}(\omega-\pi)\Big[\hat{h}_0(\omega)\hat{\tilde{g}}_1(\omega-\pi)-\hat{h}_1(\omega)\hat{\tilde{g}}_0(\omega-\pi)\Big]\\
&=\hat{x}(\omega-\pi)\text{.}
\end{align*}
\end{remark}

We also obtain a filterbank interpretation of the \emph{aliasing induced by subsampling}, in analogy to the Fourier decomposition of \eqref{eqn:FFT_aliasing}. As shown in Figure \ref{fig:filterbank03}, filterbank coefficients corresponding to the subsampled signal
$x_s[n]$ are arithmetic averages of low- and high-frequency coefficients
corresponding to $x[n]$,
in analogy to the symmetry about $\pi/2$ visible in the Figure
\ref{fig:sampling}(f). 
\begin{figure}
\centering
\includegraphics[scale=.25]{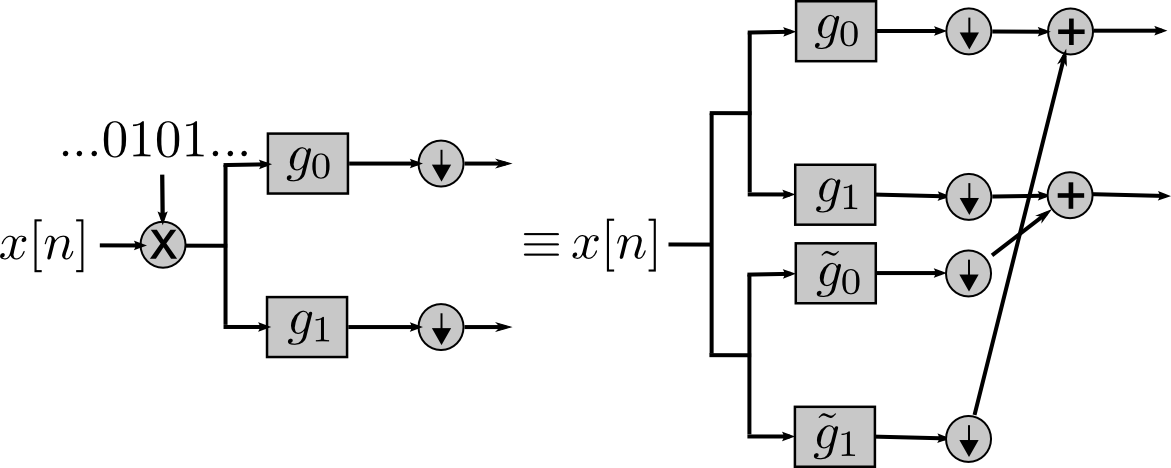}
\caption{Illustration of Corollary \ref{cor:subsampled}, with the left
  side showing $x[n]$ subject to subsampling and the right side
  showing the corresponding aliasing structure. Filterbank
  coefficients corresponding to the subsampled signal are arithmetic
  averages of \emph{complementary} low- and high-frequency coefficients.}
\label{fig:filterbank03}
\end{figure}

\begin{corollary}[Localized Aliasing]\label{cor:subsampled}
If the set $\{g_i,h_i\}_{i\in\mathbb{Z}_2}$ is a perfect
reconstruction filterbank, then by linearity and Theorem \ref{thm:reverse-order},
\begin{align*}
v^{x_s}_i[n]=\frac{1}{2}\Big(v^x_i[n]+(-1)^iw^x_{1-i}[n]\Big)\text{.}
\end{align*} 
Furthermore, this modulation implies the filterbank subband symmetry 
$v^{x_s}_i[n]=(-1)^iw^{x_s}_{1-i}[n]$.
\end{corollary}

As can be seen in Figure \ref{fig:filterbank03}, \emph{localized aliasing} occurs when $v^x_i[n]$ and $w^{x}_{1-i}[n]$ are both
simultaneously nonzero and hence indistinguishable in
$v^{x_s}_i[n]$. Unlike the global Fourier aliasing illustrated in Figure \ref{fig:sampling}, however, this aliasing is confined to a temporally localized region that depends on the local regularity of $x$.

\subsection{Extension to Multi-Level Setting}
\label{sec:multi-level}

Multi-level filterbank analysis corresponds to a recursive application
of convolution and downsampling operators to successive sets of filterbank coefficients.
To index the corresponding subbands, we adopt \emph{binary} vector notation for indices as follows. Let $\bm{i}=(i_{I\text{-}1},\dots,i_1,i_0)^T\in{\mathbb{Z}_2}^I$ and
$\bm{i}'=(i_{I\text{-}1},\dots,i_1,1-i_0)^T$, and  
recalling $v^x_i[n]$ and $w^x_i[n]$ from Definitions \ref{def:v} and
\ref{def:w}, define the corresponding $I$-level recursions:
\begin{align}\label{eqn:multi-level_v}
\begin{split}
v^x_{\bm{i}}[n]:=&\Big(g_{i_{I\text{-}1}}\star_{m_{I\text{-}1}}\Big(\dots\Big(g_{i_1}[m_1]\star_{m_1}\Big(g_{i_0}[m_0]\\
&\star_{m_0} x[m_0]\Big)[2m_1]\Big)[2m_2]\dots\Big)[2m_{I\text{-}1}]\Big)[2n]
\end{split}\\
\label{eqn:multi-level_w}
\begin{split}
w^x_{\bm{i}}[n]:=&\Big(g_{i_{I\text{-}1}}\star_{m_{I\text{-}1}}\Big(\dots\Big(g_{i_1}[m_1]\star_{m_1}\Big(\tilde{g}_{i_0}[m_0]\\
&\star_{m_0} x[m_0]\Big)[2m_1]\Big)[2m_2]\dots\Big)[2m_{I\text{-}1}]\Big)[2n]\text{.}
\end{split}
\end{align}
Here $i_k\in\mathbb{Z}_2$ indexes the
analysis filters used in the $k$th-level decomposition (i.e.~$g_0$ or
$g_1$), and $\bm{i}'$ corresponds to a high (low) frequency subband when
$\bm{i}$ is a low (high) frequency subband.
Note here that the \emph{complementary} filters
$\tilde{g}_0$ and $\tilde{g}_1$ are used \emph{only} in the
$0$th-level decomposition in \eqref{eqn:multi-level_w}.
The corresponding perfect reconstruction extension of
\eqref{eqn:inverse_filterbank} to the case of an $I$-level filterbank is
\begin{align}\label{eqn:inverse_filterbank_multi-level}
\hat{x}_r(\omega):=\sum_{\bm{i}}\hat{v}^x_{\bm{i}}(2^I\omega)\prod_{k=0}^{I-1}\hat{h}_{i_k}(2^k\omega)\text{.}
\end{align}
Then, in parallel to our earlier development, 
the results of Theorem \ref{thm:reverse-order} and Corollary \ref{cor:subsampled} extend to the multi-level setting as follows.

\begin{theorem}[Multi-Level ROSS]\label{thm:multi-level}
Suppose the set $\{g_i, h_i\}_{i\in\mathbb{Z}_2}$ is a perfect
reconstruction filterbank, and let $\bm{i}\in{\mathbb{Z}_2}^I$. Then,
\begin{align*}
v^{x_m}_{\bm{i}}=&(-1)^{i_0}w^x_{\bm{i}'}[n];\\
v^{x_s}_{\bm{i}}=&(-1)^{i_0}w^{x_s}_{\bm{i}'}[n]=\frac{1}{2}\Big(v^x_{\bm{i}}[n]+(-1)^{i_0}w^x_{\bm{i}'}[n]\Big)\text{.}
\end{align*}
\end{theorem}
Using the same example shown in Figure \ref{fig:sampling}, Figure \ref{fig:sampling_ROSS} illustrates this localized aliasing in the multi-level filterbank setting.
Note that although the subsampled example signal is subject to aliasing in a global sense (Figure~\ref{fig:sampling}(f)), the corresponding $v^x_{\bm{i}}[n]$ may be recovered from $v^{x_s}_{\bm{i}}[n]$ whenever $w^x_{\bm{i}'}[n]=0$. Theorem \ref{thm:multi-level} simplifies when self-complimentarity is taken into account, illustrated also in Figure \ref{fig:filterbank04}.

\begin{corollary}[Multi-Level Self-Complementary ROSS]\label{cor:multi-level_haar}
If the set $\{g_i,h_i\}_{i\in\mathbb{Z}_2}$ is a perfect
reconstruction filterbank that is also self-complementary, then
\begin{align*}
v^x_{\bm{i}}[n]=&(-1)^{1-i_0}w^x_{\bm{i}}[n]\text{,}\\
v^{x_m}_{\bm{i}}[n]=&v^x_{\bm{i}'}[n]\text{;}\\
v^{x_s}_{\bm{i}}[n]=&\frac{1}{2}\Big(v^x_{\bm{i}}[n]+v^{x_m}_{\bm{i}}[n]\Big)=\frac{1}{2}\Big(v^x_{\bm{i}}[n]+v^x_{\bm{i}'}[n]\Big)\text{.}
\end{align*}
\end{corollary}

\begin{figure}
\centering
\subfigure[$v^x_{\bm{i}}$]{\includegraphics[width=.075\textwidth]{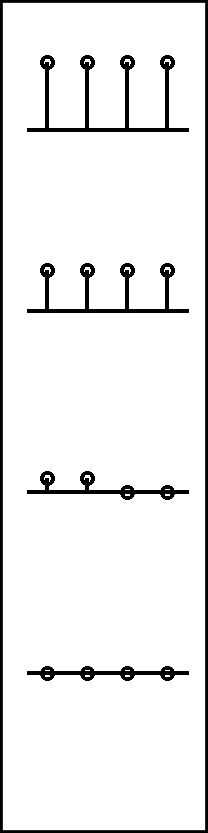}}
\subfigure[$v^{x_m}_{\bm{i}}$]{\includegraphics[width=.075\textwidth]{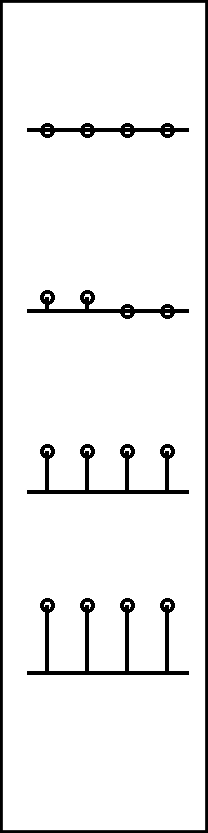}}
\subfigure[$v^{x_s}_{\bm{i}}$]{\includegraphics[width=.075\textwidth]{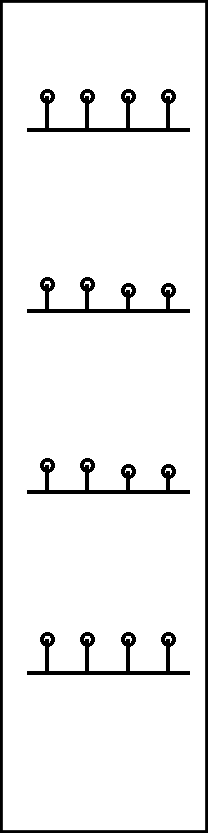}}
\subfigure[$w^{x}_{\bm{i}}$]{\includegraphics[width=.075\textwidth]{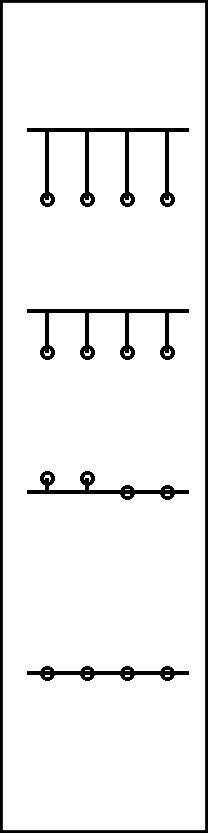}}
\caption{Pictorial illustration of \emph{localized aliasing} in the 2-level filterbank domain, as indicated by Theorem \ref{thm:multi-level}. Parts (a-c) show filterbank coefficients sequences corresponding to Fig.~\ref{fig:sampling}(a-c), respectively, with (d) the complementary sequence corresponding to Fig.~\ref{fig:sampling}(a). From top to bottom, the ordering of the four subbands represented in each subfigure is $\bm{i}=(0,0),(1,0),(1,1),(0,1)$. ``Rewiring'' is evident in comparing (b) with (d), and $v^x_{\bm{i}}[n]$ is exactly recoverable from $v^{x_s}_{\bm{i}}[n]$ whenever $w^x_{\bm{i}'}[n]=0$.}
\label{fig:sampling_ROSS}
\end{figure}

\begin{figure*}
\centering
\includegraphics[scale=.25]{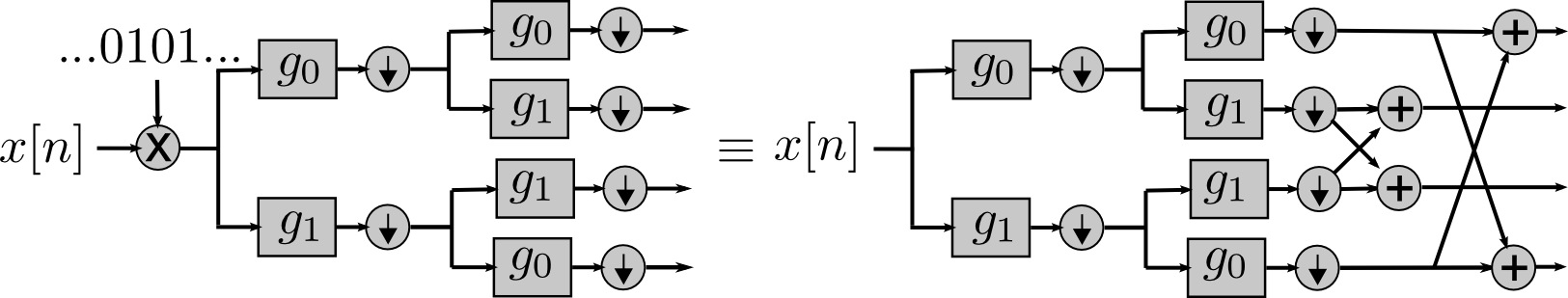}
\caption{Illustration of Corollary \ref{cor:multi-level_haar}. Filterbank
  coefficients corresponding to the subsampled signal are arithmetic
  averages of reverse-order coefficients. 
\emph{Localized aliasing} thus occurs  
when $v^x_{\bm{i}}[n]$ and $w^{x}_{\bm{i}'}[n]$ are both supported
simultaneously and hence indistinguishable in
$v^{x_s}_i[n]$.}
\label{fig:filterbank04}
\end{figure*}

\begin{remark}[Extension to Discrete Wavelet Transform]  The above results can easily be adapted to discrete
wavelet transforms, which employ the same fundamental building blocks as perfect
reconstruction filterbanks \cite{ref:Mallat_1999,ref:Strang_1996}.
As an example, the filterbank rewiring associated with the wavelet transform of a subsampled signal $x_s[n]$ is illustrated in Figure \ref{fig:wavelet}; we leave the details as an easy exercise for the reader.
\end{remark}

\begin{figure*}
\centering
\includegraphics[scale=.25]{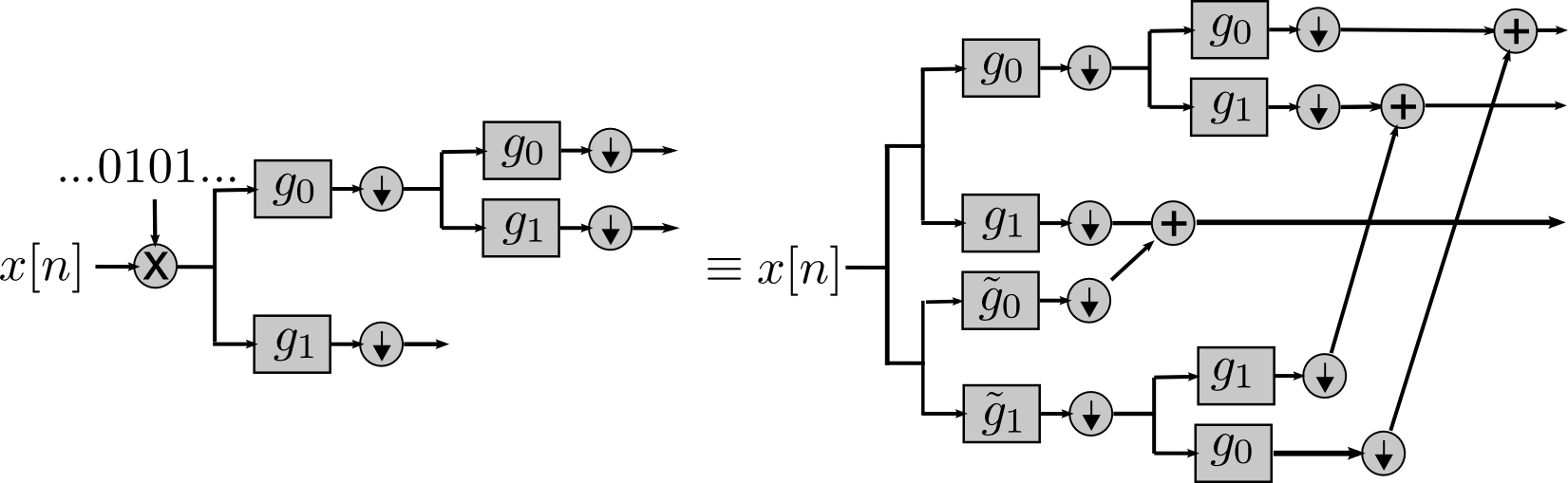}
\caption{An example of ROSS-based \emph{wavelet} analysis, with the left
  side showing $x[n]$ subject to subsampling, and the right side
  showing the corresponding aliasing structure.}
\label{fig:wavelet}
\end{figure*}

\section{Subband Convolution Structure and Localized Modulation}
\label{sec:SCS}

In the previous section, we introduced and studied the reverse-ordering of subband structure induced
by subsampling, for the case of general perfect
reconstruction filterbanks. Furthermore, we saw that this structure simplified considerably when self-complimentarity was taken into account. In this
section, we show that the symmetry of the Haar filterbank transform
also affords a characterization of Fourier group duality and
convolution. As a special case, we recover the multi-level
ROSS structure of Corollary \ref{cor:multi-level_haar}, thereby linking the ROSS results of Section \ref{sec:ROSS} and the subband convolution structure (SCS) we introduce below. 

\subsection{Subband Convolution Structure}
Recall the multi-level filterbank decomposition of \eqref{eqn:multi-level_v},
in which $i_k\in{\mathbb{Z}_2}$ indexes pairs of analysis filters used for the
$k$th level decomposition. Below, we prove the duality of time-domain 
multiplication and \emph{subband convolution} of filterbank coefficients in this context. 

\begin{theorem}[Subband Convolution]\label{thm:modulation_HFT}
Let $v^x_{\bm{i}},v^y_{\bm{i}}$ 
be $I$-level Haar filterbank coefficient sequences corresponding to
$x$ and $y$, respectively, with $v_{\bm{i}}^{xy}[n]$ that of the element-wise product $xy$. Then, letting $\circledast$ denote cyclic
convolution, we have the relation
\begin{align*}
\Big(v^x_{\bm{j}}\circledast_{\bm{j}}v^y_{\bm{j}}\Big)_{\bm{i}}[n] :=\frac{1}{2^I}\sum_{\bm{j}}v^x_{\bm{i}+\bm{j}}[n]v^y_{\bm{j}}[n]=v^{xy}_{\bm{i}}[n]\text{.}
\end{align*}
\end{theorem}

\begin{proof}
By the Fourier representation of \eqref{eqn:inverse_filterbank_multi-level}, the product $x[n]\cdot y[n]$ can be analyzed as:
\begin{align}\label{eqn:xy}
&\widehat{(x\cdot y)}(\omega)=\hat{x}_r(\omega)\star_{\omega}\hat{y}_r(\omega)\notag\\
&=\left[\sum_{\bm{i}}\hat{v}^x_{\bm{i}}(2^I\omega)\prod_{k=0}^{I-1}\hat{h}_{i_k}(2^k\omega)\right]\star_{\omega}\left[\sum_{\bm{j}}\hat{v}^x_{\bm{j}}(2^I\omega)\prod_{k=0}^{I-1}\hat{h}_{j_k}(2^k\omega)\right]\notag\\
\begin{split}
&=\sum_{\bm{i}}\sum_{\bm{j}}\int_{-\pi}^{\pi}\hat{v}^x_{\bm{i}}(2^I\nu)\hat{v}^x_{\bm{j}}\Big(2^I(\omega-\nu)\Big)\\
&\quad\quad\cdot\left[\prod_{k=0}^{I-1}\hat{h}_{i_k}(2^k\nu)\hat{h}_{j_k}\Big(2^k(\omega-\nu)\Big)\right]d\nu
\text{.}\end{split}
\end{align}
It follows from the definition of the Haar filterbank transform in \eqref{eqn:HFT} that
\begin{align*}
&\hat{h}_{i_k}(2^k\nu)\hat{h}_{j_k}\Big(2^k(\omega-\nu)\Big)\\
&=\frac{1}{4}\left[(-1)^{i_k}+e^{-j2^k\nu}\right]\left[(-1)^{j_k}+e^{-j2^k(\omega-\nu)}\right]\\
&=\frac{1}{2}\left[\hat{h}_{i_k+j_k}(2^k\omega)+(-1)^{i_k}e^{-j2^k\nu}\hat{h}_{i_k+j_k}\left(2^k(\omega-2\nu)\right)\right]\text{.}
\end{align*}
Substituting this expression into \eqref{eqn:xy} for $k=0$,
\begin{align*}
&\widehat{(x\cdot y)}(\omega)
=\frac{1}{2}\sum_{\bm{i}}\sum_{\bm{j}}\int_{-\pi}^{\pi}\hat{v}^x_{\bm{i}}(2^I\nu)\hat{v}^x_{\bm{j}}\\
&\quad\quad\cdot 
\Big(2^I(\omega-\nu)\Big)\left[\prod_{k=1}^{I-1}\hat{h}_{i_k}(2^k\nu)\hat{h}_{j_k}\Big(2^k(\omega-\nu)\Big)\right]\\
&\quad\quad\cdot
\left[\hat{h}_{i_0+j_0}(\omega)+(-1)^{i_0}e^{-j\nu}\hat{h}_{i_0+j_0}(\omega-2\nu)\right]d\nu\\
&\quad=\frac{1}{2}\sum_{\bm{i}}\sum_{\bm{j}}\hat{h}_{i_0+j_0}(\omega)
\int_{-\pi}^{\pi}\hat{v}^x_{\bm{i}}(2^I\nu)\hat{v}^x_{\bm{j}}\Big(2^I(\omega-\nu)\Big)\\
&\quad\quad\cdot\left[\prod_{k=1}^{I-1}\hat{h}_{i_k}(2^k\nu)\hat{h}_{j_k}\Big(2^k(\omega-\nu)\Big)\right]d\nu\text{,}
\end{align*}
where we have used the fact that 
$\int_{-\pi}^{\pi} e^{-j2^k\nu}\hat{f}(\nu)d\nu=0$ for all $\hat{f}(\nu)\in L(\mathbb{R}/2^{-k}\pi)$ whenever $k\geq 0$, as $e^{-j(2^k\nu-\pi)}\hat{f}(\nu-2^{-k}\pi)=-e^{-j2^k\nu}\hat{f}(\nu)$. By recursion over $k$, the above reduces to
\begin{align}\label{eqn:recursion}
\begin{split}
\widehat{(x\cdot y)}(\omega)
=&\frac{1}{2^{K}}\sum_{\bm{i}}\sum_{\bm{j}}\left[\prod_{k=0}^{K-1}\hat{h}_{i_k+j_k}(2^k\omega)\right]\\
&\cdot\int_{-\pi}^{\pi}\hat{v}^x_{\bm{i}}(2^I\nu)\hat{v}^x_{\bm{j}}\Big(2^I(\omega-\nu)\Big)\\
&\cdot\left[\prod_{k=K}^{I-1}\hat{h}_{i_k}(2^k\nu)\hat{h}_{j_k}\Big(2^k(\omega-\nu)\Big)\right]d\nu\\
=&\frac{1}{2^{K+1}}\sum_{\bm{i}}\sum_{\bm{j}}\left[\prod_{k=0}^{K}\hat{h}_{i_k+j_k}(2^k\omega)\right]\\
&\cdot\int_{-\pi}^{\pi}\hat{v}^x_{\bm{i}}(2^I\nu)\hat{v}^x_{\bm{j}}\Big(2^I(\omega-\nu)\Big)\\
&\cdot\left[\prod_{k=K+1}^{I-1}\hat{h}_{i_k}(2^k\nu)\hat{h}_{j_k}\Big(2^k(\omega-\nu)\Big)\right]d\nu\\
=\dots=&\frac{1}{2^I}\sum_{\bm{i}}\sum_{\bm{j}}\left[\prod_{k=0}^{I-1}\hat{h}_{i_k+j_k}(2^k\omega)\right]\\
&\cdot\left[\hat{v}^x_{\bm{i}}(2^I\omega)\star_{\omega}\hat{v}^y_{\bm{j}}(2^I\omega)\right]\text{.}
\end{split}
\end{align}
Note that \eqref{eqn:recursion} takes the form of a multi-level
inverse filterbank transform, as per \eqref{eqn:inverse_filterbank_multi-level}. As the Haar filterbank transform is one-to-one and onto, $v_{\bm{i}}^{xy}$ is thus uniquely defined by
\begin{align*}
\hat{v}^{xy}_{\bm{i}}(\omega)=\frac{1}{2^I}\sum_{\bm{j}}\left[\hat{v}^x_{\bm{i}+\bm{j}}(\omega)\star_{\omega}\hat{v}^y_{\bm{j}}(\omega)\right]\text{,}
\end{align*}
which agrees with the claim of the theorem.
\end{proof}

Figure \ref{fig:filterbank05} illustrates the
corresponding ``rewiring''
of filterbank subbands, in which  $v^x_{\bm{i}}[n]$ and $v^y_{\bm{i}}[n]$ are coupled together
to yield $v^{xy}_0[n]$, and $v^x_{\bm{i}'}[n]$ and $v^y_{\bm{i}}[n]$ are combined
to produce $v^{xy}_1$.
\begin{figure}
\centering
\includegraphics[scale=.25]{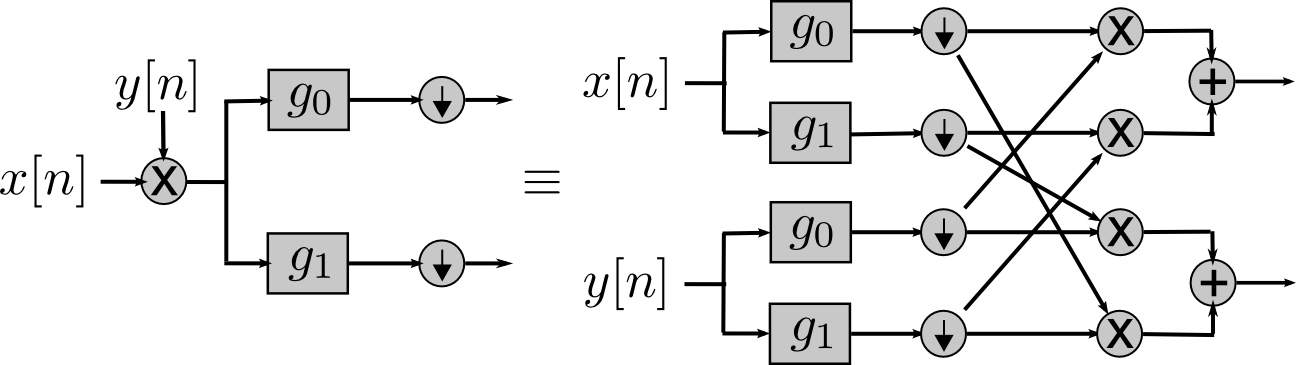}
\caption{Illustration of Theorem
  \ref{thm:modulation_HFT}, showing the correspondence between time-domain multiplication and ``logical convolution'' of filterbank subbands.}
\label{fig:filterbank05}
\end{figure}

\begin{remark}[Logical Convolution]
By restricting $x$ and $y$ be finite-dimensional, we recover the so-called logical
convolution theorem \cite{ref:Kennett_1970} as a
special case of Theorem \ref{thm:modulation_HFT}.
This is easily seen by considering an
order-$2^I$ Walsh sequence
$\vec{\phi}_{\bm{i}}\in\mathbb{R}^{2^I}$ and its Abelian structure: 
\begin{align}\label{eqn:abelian}
\text{diag}(\vec{\phi}_{\bm{i}}) \vec{\phi}_{\bm{j}}=\vec{\phi}_{\bm{i}+\bm{j}},
\quad \bm{i},\bm{j}\in\{\mathbb{Z}_2\}^I\text{.}
\end{align}
Orthogonality of the Walsh basis sets implies that any $\vec{x} \in \mathbb{R}^{2^I}$ can be expanded in terms of its Walsh-Hadamard coefficients $\Big\langle\vec{\phi}_{\bm{i}},\vec{x}\Big\rangle=\vec{\phi}_{\bm{i}}^T\vec{x}$
as $\vec{x}=2^{-I}\sum_{\bm{j}}\Big\langle\vec{\phi}_{\bm{j}},\vec{x}\Big\rangle\vec{\phi}_{\bm{j}}$,
and hence the group homomorphism of \eqref{eqn:abelian} yields the desired relation for all
$\vec{x},\vec{y}\in \mathbb{R}^{2^I}$:
\begin{align*}
\Big\langle\vec{\phi}_{\bm{i}},\text{diag}(\vec{x})\vec{y}\Big\rangle =& \frac{1}{2^{2I}}\sum_{\bm{j},\bm{j}'} \Big\langle\vec{\phi}_{\bm{j}},\vec{x}\Big\rangle 
\Big\langle \vec{\phi}_{\bm{j}'},\vec{y}\Big\rangle
\Big\langle\vec{\phi}_{\bm{i}},
\text{diag}(\vec{\phi}_{\bm{j}}) \vec{\phi}_{\bm{j}'} 
\Big\rangle\\
=&
\frac{1}{2^I}\sum_{\bm{j}}\Big\langle\vec{\phi}_{\bm{j}},\vec{x}\Big\rangle\Big\langle\vec{\phi}_{\bm{i}+\bm{j}},\vec{y}\Big\rangle\text{.}
\end{align*}
\end{remark}

\subsection{Localized Modulation and Connection to Reverse-Ordered Subband Structure}

We conclude this section by interpreting the result of Theorem \ref{thm:modulation_HFT} 
in terms of amplitude modulation and sampling. 

\begin{remark}[Multi-Level ROSS for Haar Filterbank Transform]
Suppose that we set $v^y_{(0,\dots,0,0)}[n]=v^y_{(0,\dots,0,1)}[n]=2^{I-1}$ for all $n$, and
$v^y_{\bm{i}}[n]=0$ otherwise, thus yielding a Dirac comb. Then
Theorem \ref{thm:modulation_HFT} agrees precisely with Corollary \ref{cor:multi-level_haar}:
\begin{align*}
\Big(v^x_{\bm{j}}\circledast_{\bm{j}}v^y_{\bm{j}}\Big)_{\bm{i}}[n]
=&\frac{1}{2}\Big(v^x_{\bm{i}}[n]+v^x_{\bm{i}'}[n]\Big)
\text{.}
\end{align*}
\end{remark}
\begin{remark}[Generalized Subsampling]
More generally, suppose
$y[n]\in\{0,1\}$ is a sampling mask of any kind. Then it may be seen from 
Figure \ref{fig:filterbank05} that the subsampled
signal $x[n]y[n]$ is aliased if $v^x_0[n]v^y_0[n]$ and
$v^x_1[n]v^y_1[n]$ (or, $v^x_0[n]v^y_1[n]$ and $v^x_0[n]v^y_1[n]$) are
simultaneously supported. 
\end{remark}

\begin{remark}[Localized Modulation]
Consider Figure \ref{fig:filterbank05} again and suppose at time
$n=0$, we have that $v^y_0[n]=0$ and $v^y_1[n]=1$.
Then $v^{x}_0[n]$ is ``modulated'' to $v^{xy}_1[n]$.
This is similar to Fourier amplitude modulation, in which case the energy of the
modulated signal is concentrated around a chosen carrier frequency;
however, a representation based on
Theorem \ref{thm:modulation_HFT} is amenable to temporally local
processing. For example, suppose at time $n=1$, we take $v^y_0[n]=1$ and
$v^y_1[n]=0$; then $v^{x}_0[n]$ is mapped to $v^{xy}_0[n]$ instead of
$v^{xy}_1[n]$. In other words, this filterbank interpretation of
\emph{localized modulation}---illustrated in Figure \ref{fig:SCS}---is ideal for tracking modulation when the
``carrier'' $y[n]$ is allowed to change over time. 
\end{remark}

Formally, we are concerned with characterizing a sum 
 $z[n]=\sum_k x_k[n]y_k[n]$ of modulated sequences $x_k$. When the ``envelope function''
$y_k[n]$ is chosen carefully, then it follows from
Theorem~\ref{thm:modulation_HFT} that these signals are recoverable.

\begin{proposition}[Localized Amplitude Modulation]
Suppose $y_k[n]$ is defined by a sequence of index values $\bm{j}_k\in\{\mathbb{Z}_2\}^I$
and Haar filterbank coefficients: $v^{y_k}_{\bm{i}}[n]=\delta(\bm{i},\bm{j}_k[n])$. 
Then $x_k[n]$ is recoverable from
$z[n]=\sum_k x_k[n]y_k[n]$ when the supports of $v^{x_k}_{\bm{i}+\bm{j}_k[n]}[n]$ are
mutually exclusive for all $\bm{i}$ and $k$.
\end{proposition}
\begin{proof}
By the subband convolution result of Theorem \ref{thm:modulation_HFT},
\begin{align*}
v^{z}_{\bm{i}}[n]=&\sum_k v^{x_ky_k}_{\bm{i}}[n]
=\sum_k\Big(v^{x_k}_{\bm{j}}\circledast_{\bm{j}}v^{y_k}_{\bm{j}}\Big)_{\bm{i}}[n]
=\sum_kv^{x_k}_{\bm{i}+\bm{j}_k[n]}[n]
\end{align*}
However, by the assumed mutual exclusivity of the supports of $v^{x_k}_{\bm{i}+\bm{j}_k[n]}[n]$, it follows that
\begin{align*}
v^z_{\bm{i}}[n]=
\left\{
\begin{array}{ll}
v^{x_k}_{\bm{i}+\bm{j}_k[n]}[n]&\text{if there exists $k$ 
with nonzero $v^{x_k}_{\bm{i}+\bm{j}_k[n]}$}\\
0&\text{otherwise,}
\end{array}
\right.
\end{align*}
and thus we conclude that $v^{x_k}_{\bm{i}}[n]=v^{z}_{\bm{i}+\bm{j}_k[n]}[n]$ whenever
$v^{v_k}_{\bm{i}}[n]$ is nonzero.
\end{proof}

\begin{figure}
\centering
\subfigure[$v^x_{\bm{i}}$]{\includegraphics[width=.075\textwidth]{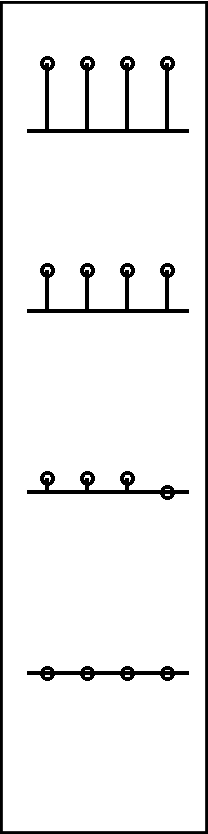}}
\subfigure[$v^y_{\bm{i}}$]{\includegraphics[width=.075\textwidth]{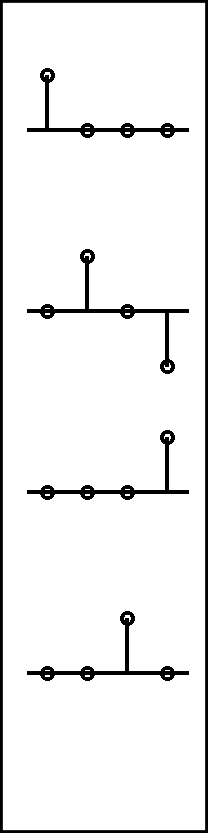}}
\subfigure[$v^{xy}_{\bm{i}}$]{\includegraphics[width=.075\textwidth]{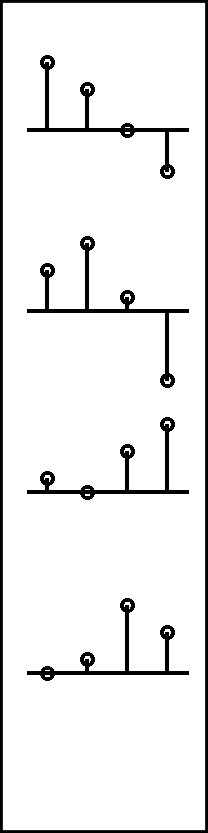}}
\caption{Pictorial illustration of \emph{localized modulation} in 2-level filterbank domain. Time-domain multiplication of $x$ and $y$, represented by their filterbank coefficients in (a) and (b), results in subband convolution, shown in (c); ``rewiring'' is evident in comparing (a) and (c). Exact recovery of $v^x_{\bm{i}}[n]$ from $v^{xy}_{\bm{i}}[n]$ is possible when the supports of $v^x_{\bm{i}+\bm{j}}[n]$ are mutually exclusive for all $\bm{i}$ and $\bm{j}$ corresponding to nonzero $v^x_{\bm{i}}[n]$ and $v^y_{\bm{j}}[n]$.
}
\label{fig:SCS}
\end{figure}
                                     
\section{Discussion: Implications for Signal Analysis}
\label{sec:likelihood}

The preceding two sections have explored properties of reverse-order and convolution subband structure (ROSS and SCS) in filterbanks, and shown their relation to the concepts of localized aliasing and modulation.  We now discuss the practical implications of these results for signal analysis, and provide two brief demonstrations that ``rewiring'' filterbank diagrams in the manner of ROSS and SCS can enable new solutions to problems involving subsampled data corrupted by additive or multiplicative noise.  As many scientific and engineering applications give rise to inverse problems involving subsampled and/or noisy data, and since filterbanks are the tool of choice for many signal and image processing tasks, it is natural to analyze the data likelihoods resulting from these problems directly in the filterbank coefficient domain.

\subsection{Filterbank-Domain Likelihoods via ROSS and SCS}

While data likelihoods often do not admit straightforward closed-form expressions through traditional filterbank analysis, the ROSS and SCS concepts provide a new way to characterize signals subject to aliasing or signal-dependent noise effects directly in the filterbank coefficient domain, by virtue of the associated filterbank rewiring techniques. The expression of filterbank data likelihoods is key to solving signal reconstruction and enhancement problems in this context; as signal acquisition models, these likelihoods may be coupled with regularization terms that encourages parsimony as a means of signal modeling, reflected through prior probability densities on filterbank coefficients, or equivalently through terms that explicitly penalize complexity.  To this end, the following two corollaries show how explicit likelihood formulations follow directly from application of the ROSS and SCS concepts introduced earlier.

\begin{corollary}[Noisy, Subsampled Data Likelihood]\label{cor:likelihood1}
Fix $x \in \ell^2(\mathbb{Z})$ as the signal of interest, and let $\xi$ comprise samples of white Gaussian noise of variance $\sigma^2$.  Suppose then that we observe \emph{subsampled, noisy} data $y = x_s + \xi_s$ and subsequently apply a \emph{unitary} filterbank transform; then it follows from the localized aliasing relation of Theorem \ref{thm:multi-level} that the analysis filterbank coefficients of $y$ satisfy 
\begin{align*}
v^y_{\bm{i}}[n]=&\frac{v^x_{\bm{i}}[n]+(-1)^{i_0}w^x_{\bm{i}'}[n]}{2}+\frac{v^{\xi}_{\bm{i}}[n]+(-1)^{i_0}w^{\xi}_{\bm{i}'}[n]}{2}\text{,}
\end{align*}
and hence each admits a Normal likelihood with mean $\big(v^x_{\bm{i}}[n]+(-1)^{i_0}w^x_{\bm{i}'}[n]\big)/2$ and variance $\sigma^2/2$.
\end{corollary}

\begin{corollary}[Multiplicative Noise Data Likelihood]\label{cor:likelihood2}
 Suppose instead that the observation model $y[n] = x[n] + x[n] \xi[n]$ is in force.  Then the subband convolution structure of Theorem \ref{thm:modulation_HFT} implies the filterbank coefficient relation 
\begin{align*}
v^y_{\bm{i}}[n]=&v^x_{\bm{i}}[n]+\Big(v^x_{\bm{j}}[n]\circledast_{\bm{j}}v^{\xi}_{\bm{j}}[n]\Big)_{\bm{i}}\text{,}
\end{align*}
and hence the likelihood form of $v^y_{\bm{i}}[n]$ is multivariate Normal with mean $v^x_{\bm{i}}[n]$, where the covariance of $v^y_{\bm{i}}[n]$ and $v^y_{\bm{j}}[n]$ is given by $\sigma^2(v^x_{\bm{k}}\circledast_{\bm{k}}v^x_{\bm{k}})_{\bm{i}+\bm{j}}$.
\end{corollary}

\subsection{Proof of Concept: Application to Image Interpolation and Denoising}

The likelihood expressions of Corollaries \ref{cor:likelihood1} and \ref{cor:likelihood2} above extend straightforwardly to the case of separable two-dimensional filterbank transforms, which in turn are typically employed in imaging applications.  Thus, to illustrate the practical applicability of these results, we now undertake two proof-of-concept experiments that are representative of problems frequently encountered in digital imaging, and for which our ROSS and SCS characterizations---in contrast to standard approaches---yield closed-form likelihood expressions for the corresponding filterbank coefficients.

The first of these---image interpolation in the presence of noise---is made difficult by the fact that low- and high-frequency filterbank subbands interact with one another as well as with the noise itself; the ``rewiring'' expression of Corollary \ref{cor:likelihood1} in turn provides a closed-form likelihood expression for the filterbank coefficients.  In the second experiment, we consider the similarly difficult problem of mitigating \emph{multiplicative} noise; in this case, Corollary \ref{cor:likelihood2} yields the corresponding likelihood.

As filterbank coefficients of images typically exhibit sparsity\cite{ref:Crouse_1998}, one natural approach to utilize these likelihoods through a Bayesian framework, in which transform coefficients of the underlying image $x$ are modeled as random variables $v_{\bm{i}}^X[n]$ taking zero-mean, symmetric, and unimodal ``heavy-tailed'' distributions that exhibit super-Gaussian tail behavior\cite{ref:Portilla_2003,ref:Selesnick_2006}.  We assume such a prior distribution here, and evaluate posterior means numerically via Monte Carlo averages.

We follow typical practice in approximating the overall joint posterior distribution of all filterbank coefficients by a product of marginal distributions, where each marginal posterior is associated with a particular subband.  In turn, the $\ell^2$-optimal estimator of filterbank coefficients $v_{\bm{i}}^X[n]$ corresponding to the $\bm{i}$th subband given the corresponding data coefficients $v_{\bm{i}}^y[n]$ is given by
\begin{align}\label{eqn:mmse}
\mathbb{E}[v^X_{\bm{i}}|v^y_{\bm{i}}]
=\frac{\int v^x_{\bm{i}} p(v^y_{\bm{i}}|v^x_{\bm{i}})p(v^x_{\bm{i}})dv^x_{\bm{i}}}{p(v^y_{\bm{i}})}\text{;}
\end{align}
the corresponding synthesis filterbank in turn allows reconstruction of the estimated image.

We first consider a well-known 8-bit test image that has been artificially downsampled and degraded with additive white Gaussian noise of variance 400 to yield a signal-to-noise ratio (SNR) of 16.77~dB relative to $x_s$ (or 1.22~dB relative to $x$), as shown in Figures \ref{fig:ex01}(a) and \ref{fig:ex01}(b).  Figure \ref{fig:ex01}(c) shows the corresponding image reconstruction, which retains much of sharpness of the image edges and textures while suppressing noise, resulting in an SNR gain of 3.81~dB.  For purposes of comparison, Figure \ref{fig:ex01}(d) shows the result of the recently proposed simultaneous interpolation and denoising method of \cite{ref:Zhang_2008}, which yields an SNR gain of  3.55~dB; this reconstruction exhibits edges that are more strongly preserved, but at the expense of greater smoothing of image textures.

In our second experiment, we consider a synthetic aperture radar (SAR) image, available at www.sandia.gov/radar; as may be seen from Figure \ref{fig:ex02}(a), such images suffer from the effects of multiplicative noise \cite{ref:Achim_2003}.  Figure \ref{fig:ex02}(b) shows the enhanced image resulting from a ``rewiring'' approach, which exhibits reduced noise in smooth and textured regions, and avoids the introduction of artifacts.  In contrast, Figure \ref{fig:ex02}(c) illustrates the standard approach: application of a logarithmic transformation to the data, followed by an additive denoising technique (here the well-known method of \cite{ref:Portilla_2003}) and subsequent exponentiation of the result.  Not only are Bayes optimality properties of \cite{ref:Portilla_2003} lost in the exponentiation transformation back to the pixel domain, but also substantial artifacts are seen to result from this standard approach.

\begin{figure*}
\centering
\subfigure[Original test image]{\includegraphics[width=.45\textwidth]{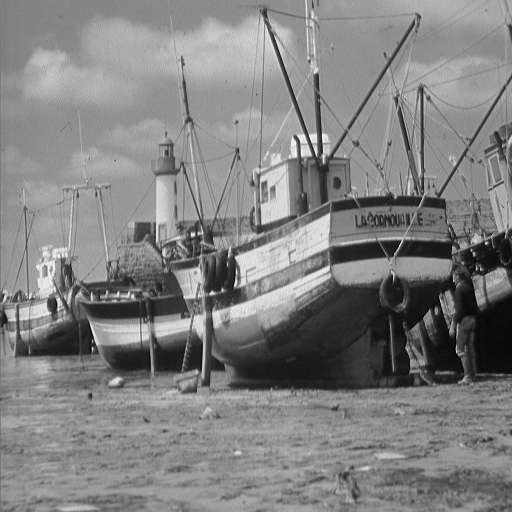}}
\subfigure[Subsampled, noisy test image]{\includegraphics[width=.45\textwidth]{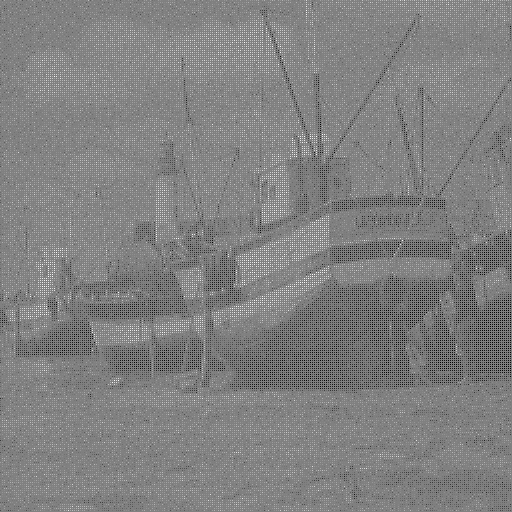}}\\
\subfigure[``Rewiring'' reconstruction]{\includegraphics[width=.45\textwidth]{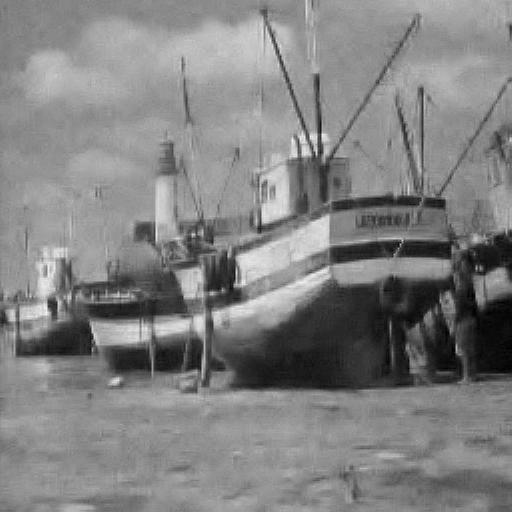}}
\subfigure[Reconstruction via the method of \cite{ref:Zhang_2008}]{\includegraphics[width=.45\textwidth]{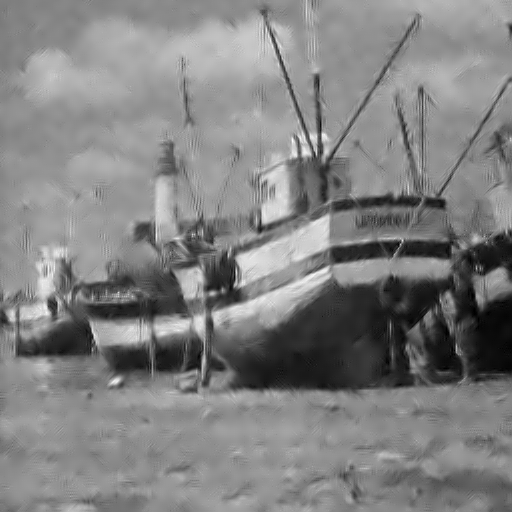}}
\caption{
Example of noisy image interpolation via the ROSS technique of Section \ref{sec:ROSS}: The first row shows a full-resolution 8-bit test image (a), along with a subsampled version that has been degraded with additive white Gaussian noise of variance 400 (b).  The bottom row shows a posterior mean reconstruction based on the filterbank-domain likelihood of Corollary \ref{cor:likelihood1} and a heavy-tailed prior distribution on filterbank coefficients (c), along with a reconstruction according to the recently proposed method of \cite{ref:Zhang_2008}, shown for comparison (d).}
\label{fig:ex01}
\end{figure*}

\begin{figure*}
\subfigure[SAR image showing ``speckled noise'']{\includegraphics[width=.33\textwidth]{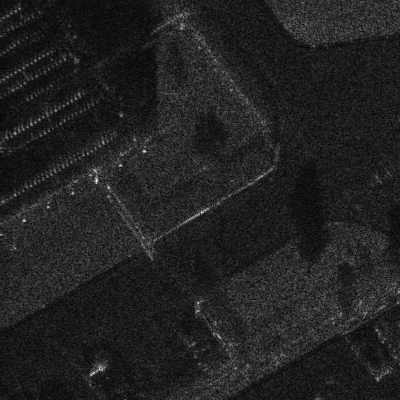}}
\subfigure[``Rewiring'' enhancement]{\includegraphics[width=.33\textwidth]{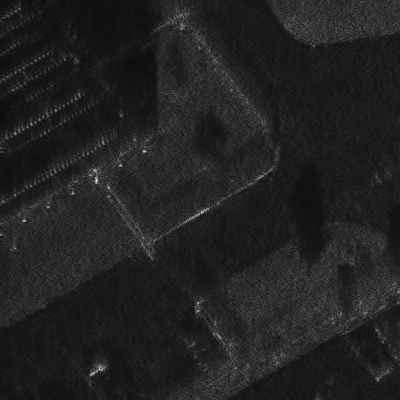}}
\subfigure[Enhancement via the method of \cite{ref:Portilla_2003}, applied after logarithmic transformation]{\includegraphics[width=.33\textwidth]{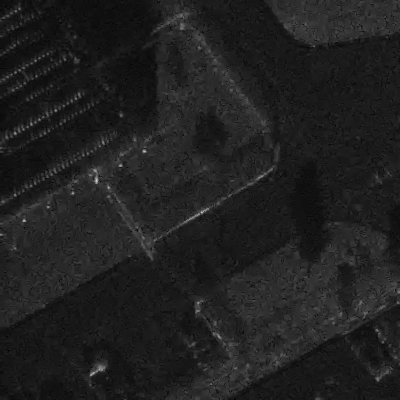}}
\caption{
Example of multiplicative noise mitigation via the SCS technique of Section \ref{sec:SCS}: Panel (a) shows a portion of SAR imagery data in which noise is visibly present.  Panel (b) shows an enhancement based on the filterbank-domain likelihood of Corollary \ref{cor:likelihood2} and a heavy-tailed prior distribution on filterbank coefficients, with panel (c) showing for comparison an enhancement according to the method of \cite{ref:Portilla_2003}, designed for additive noise and applied after a variance-stabilizing logarithmic transformation.}
\label{fig:ex02}
\end{figure*}
\subsection{Concluding Remarks}

In conclusion, we have shown in this article how filterbank ``rewirings,'' corresponding to compositions of convolution, modulation, and downsampling operators, admit expressions of \emph{localized} aliasing and modulation, in directly analogy to the global setting of Fourier analysis.  In addition to establishing a number of results that formalize reverse-order and convolution subband structures in filterbank transforms in Sections \ref{sec:ROSS} and \ref{sec:SCS}, respectively, we have demonstrated in this section how these concepts in turn enable the establishment of closed-form likelihood functions for the direct filterbank analysis of signals subject to degradations such as
missing data, spatially or temporally multiplexed data acquisition, or signal-dependent noise, such as are often encountered in practical signal processing applications.

\section*{Acknowledgment}

The authors would like to thank Sandia National Laboratories for
generously providing access to the synthetic aperture radar data; and Drs.~Lei Zhang, Xin Li, and Javier Portilla  for
kindly providing their code for the purpose of comparative evaluation.

\end{document}